\documentclass[11pt]{article}

\usepackage[colorlinks=true, linkcolor=blue, citecolor=blue, urlcolor=blue]{hyperref}
\usepackage{graphicx}
\usepackage{xcolor}
\usepackage{amsmath,amssymb,amsthm}
\usepackage{aliascnt}
\usepackage{cleveref}
\usepackage{natbib}
\usepackage{ytableau}
\usepackage[margin=1in]{geometry}

\newtheorem{thm}{Theorem}[section]
    \crefformat{thm}{\bfseries theorem~#2#1#3}
    \Crefformat{thm}{\bfseries Theorem~#2#1#3}

    \crefformat{corollary}{\bfseries corollary~#2#1#3}
    \Crefformat{corollary}{\bfseries Corollary~#2#1#3}
\newtheorem{lemma}[thm]{Lemma}
    \crefformat{lemma}{\bfseries lemma~#2#1#3}
    \Crefformat{lemma}{\bfseries Lemma~#2#1#3}
\newtheorem{prop}[thm]{Proposition}
    \crefformat{prop}{\bfseries proposition~#2#1#3}
    \Crefformat{prop}{\bfseries Proposition~#2#1#3}

\theoremstyle{definition}
\newaliascnt{defin}{thm}

    \crefformat{defin}{\bfseries definition~#2#1#3}
    \Crefformat{defin}{\bfseries Definition~#2#1#3}

    \crefformat{exa}{\bfseries example~#2#1#3}
    \Crefformat{exa}{\bfseries Example~#2#1#3}

\theoremstyle{plain}

\title{Induced Representations in Cooperative Games with Homogeneous Groups of Players}
\author{Windsor Kiang}
\date{May 2026}

\begin{document}
\maketitle

\begin{abstract}
Oftentimes, the Shapley value, a measure of the contribution of a player to a game, becomes infeasible for games with many players. However, establishing symmetry allows for polynomial-time computation. To examine this reduction, we identify the spectrum of a homogeneous group game by using an induced representation from a Young subgroup. We prove that the depth of interaction of a two-group game is limited by the size of the minority group. Therefore, the algebraic structure of the game filters out a large space of irrelevant complexities. We then show that this filtration constrains any symmetric linear value to a specific subspace. This recovers the Shapley value uniquely for games consisting of exactly two homogeneous groups under standard axioms. Finally, we explore applications to the UN Security Council and complementary goods markets to illustrate the practical power of this approach.
\end{abstract}

\noindent \textbf{Keywords:} Cooperative Game Theory, Shapley Value, Representation Theory, Symmetric Group, Homogeneous Groups.

\section{Introduction}
  While game theory studies the interactions between players, cooperative game theory focuses on understanding how payouts and costs are distributed among players who form some sort of ``coalition'', or more generally, a team. The widespread answer to this question is the Shapley Value: a computational method where each player in the game receives the average payout of their marginal contribution to all coalitions they are part of, weighted by coalition size. It is a unique solution that satisfies efficiency, symmetry, linearity, and the null player property.

  However, feasibility problems arise when the number of coalitions grows. Since the Shapley Value is meticulously calculated for $2^n$ coalitions, large $n$ values lead to overwhelming, intractable computation \citep{Hernández2015}. An interesting response, developed by \citet{Hernández2015}, is to group players of identical attributes (from a mathematical standpoint) into homogeneous classes within which players would be indistinguishable. Hernández-Lamoneda and Sánchez-Sánchez coin these games ``groups' games.'' Since the worth of any coalition then depends only on how many members it draws from each class, the number of distinct evaluations drops to a polynomial in $n$ \citep{Hernández2015}.

  In this paper, our goal is to understand the abstract underlying concepts that dictate \emph{why} this works, not just \emph{that} it works. Combinatorial arguments give little insight into the geometry of the game space and representations. \citet{Hernández2007} demonstrated what the game space looks like in the fully symmetric case: the space of transferable utility games (TU-games) $\mathcal{G}_n$ splits into three canonical invariant pieces under the action of $S_n$. Specifically, it is split into a constant space $C$, an essential space $U$, and a null space $W$. Then, any linear symmetric solution is determined entirely by its action on $C \oplus U$, vanishing on $W$.

  We now extend this idea to games with partial symmetry under a Young subgroup $H = S_{n_1} \times \dots \times S_{n_m}$, corresponding to $m$ groups of sizes $n_1, \dots, n_m$ \citep{Hernández2015}. All group-symmetric games are comprised of irreducible $H$-modules. These modules are basic building blocks which induce up to $S_n$ and can be computed through Frobenius reciprocity and the Littlewood-Richardson rule \citep{Sagan2001}. Analyzing the intersection between this induced decomposition and the aforementioned canonical subspaces $C$, $U$, and $W$ makes apparent which components of the game do and do not carry payoff information \citep{Kleinberg1985}. As part of one of our examples later on, we see that the UN Security Council has a symmetry of $S_5 \times S_{10}$, which reduces a 32{,}767-dimensional problem to one of just dimension 65. This makes it tractable to analyze the efficiency and fairness axioms on a small set of components \citep{Hernández2015}.

  The rest of this paper is organized as follows. Section 2 discusses the preliminaries, setting up the concepts of TU-games, group actions, and representation theory of $S_n$ that will be used hereafter. Section 3 continues by defining homogeneous group games (games that are invariant under a subgroup of the symmetric group) and characterizing $H$-fixed game space as a module of $H$ and of $S_n$ \citep{Hernández2015}. Section 4 then develops the restriction-induction approach where we decompose $H$-invariant games into irreducible $H$-representations and induce them to $S_n$, which can be computed via Frobenius reciprocity and Littlewood--Richardson coefficients \citep{Sagan2001}; this section then identifies how the result splits into the canonical $C$, $U$, and $W$ subspaces. Section 5 shows that symmetric linear values on these subspaces actually vanish on $W$ and act as a scalar on each irreducible piece of $C \oplus U$ \citep{Hernández2007}. Classical axioms are then re-examined through this new lens. Section 6 presents two applications: the UN Security Council and the market for complementary goods. Finally, we conclude in Section 7 with a summary and directions for future work.

\section{Preliminaries}

The following sections fix notation and collect the background needed later. We cover TU-games and the Shapley value, the action of $S_n$ on the game space, and its decomposition into irreducible subspaces.

\subsection{Cooperative TU-Games and Solutions}
Let $N = \{1, 2, \dots, n\}$ be a finite set of players. A \textit{cooperative game with transferable utility} (or TU-game) is a pair $(N, v)$, where $v: 2^N \to \mathbb{R}$ is a characteristic function satisfying $v(\varnothing) = 0$. The quantity $v(S)$ represents the \textit{worth} of the coalition $S \subseteq N$ --- its economic potential as a stand-alone number attached to that coalition. We reserve the word \textit{value} exclusively for the payoff a \textit{solution} (defined next) assigns to a \textit{player}, to avoid confusing the two throughout this paper. We write $\mathcal{G}_n$ for the vector space of all such games; since a game is uniquely determined by its worths on the $2^n - 1$ non-empty subsets of $N$, this space has dimension $2^n - 1$.

A \textit{solution}, or \textit{value}, is a map $\psi: \mathcal{G}_n \to \mathbb{R}^n$ assigning a payoff vector $\psi(v) \in \mathbb{R}^n$ to every game. We focus on \textit{linear} values, which satisfy $\psi(\alpha v + \beta w) = \alpha \psi(v) + \beta \psi(w)$. The most prominent solution concept is the \textit{Shapley value} $\phi$, defined as:
\begin{equation}
\phi_i(v) = \sum_{S \subseteq N \setminus \{i\}} \frac{|S|!(n-|S|-1)!}{n!} [v(S \cup \{i\}) - v(S)].
\end{equation}
The Shapley value is uniquely characterized by the axioms of \textit{efficiency} ($\sum \phi_i(v) = v(N)$), \textit{symmetry} (substitutable players receive equal payoffs), \textit{linearity}, and the \textit{null player property} (players who contribute zero marginal value to all coalitions receive zero payoff).

\subsection{Group Actions and Representation Theory}
Formally, a linear representation of $S_n$ on a vector space $V$ is a homomorphism from the group to $GL(V)$, determining how group elements act as linear transformations. A subspace $U' \subseteq V$ is \textit{invariant} (or a submodule) if for every $\sigma \in S_n$ and every $u \in U'$, the vector $\sigma u$ remains in $U'$.

The symmetric group $S_n$ acts on the player set $N$ by permutation, inducing an action on $2^N$ and thus a natural linear action on the space of games $\mathcal{G}_n$. For any permutation $\sigma \in S_n$ and any game $v$, the permuted game $\sigma v$ is defined by:
\begin{equation}
(\sigma v)(S) = v(\sigma^{-1} S)
\end{equation}
for all $S \subseteq N$. This makes $\mathcal{G}_n$ a linear representation of $S_n$---specifically, the permutation representation on the set $2^N \setminus \{\varnothing\}$.

A solution $\psi$ is \textit{symmetric} if it is equivariant under this action, meaning $\psi(\sigma v) = \sigma(\psi(v))$, where the action on the payoff vector is the permutation of coordinates.

Representation theory allows us to decompose $\mathcal{G}_n$ into invariant subspaces. An \textit{irreducible representation} (irrep) is an invariant subspace containing no proper invariant subspaces. The irreducible representations of $S_n$ are in one-to-one correspondence with \textit{partitions} $\lambda$ of $n$ \citep[Chapter 2]{Sagan2001}. Also, if $\lambda$ is a partition of $n$, we write $\lambda \vdash n$. We define the \textit{length} of a partition, $\ell(\lambda)$, as the number of parts (rows in the corresponding Young diagram).

In the coming notation, $V_{\lambda}$ denotes the specific irreducible representation (or Specht module) of $S_n$ indexed by $\lambda$. Specifically, $V_{[n]}$ corresponds to the single-row partition $(n)$, denoting the one-dimensional trivial representation. Similarly, $V_{[n-1,1]}$ corresponds to the partition $(n-1, 1)$, denoting the standard representation of dimension $n-1$, which captures the symmetries of vectors in $\mathbb{R}^n$ that sum to zero.

As established by \citet{Hernández2007}, the space of games decomposes into three canonical invariant subspaces:
\begin{equation}
\mathcal{G}_n = C \oplus U \oplus W
\end{equation}
where $C \cong n V_{[n]}$ is the isotypic component of the trivial representation, containing $n$ orthogonal copies of $V_{[n]}$ (one for each coalition size, corresponding to games whose worth depends only on coalition size); $U$ contains copies of the standard representation $V_{[n-1, 1]}$ (additive games); and $W$ consists of the remaining irreducible components that lie in the kernel of all symmetric linear values \citep{Kleinberg1985}.

\subsection{Homogeneous Groups and the Profile Lattice}
We now consider the specific structure of games with homogeneous players. Let $N$ be partitioned into $m$ disjoint groups $T_{1},\dots,T_{m}$ with sizes $n_{1},\dots,n_{m}$, where $\sum n_{j}=n$. We define the group-symmetry subgroup $H\subseteq S_{n}$ as the direct product of the symmetric groups acting on each part of the partition:
\begin{equation}
H=S_{n_{1}}\times S_{n_{2}}\times\cdots\times S_{n_{m}}
\end{equation}
This is known as a Young subgroup \citep{Sagan2001}, corresponding to the partition $\mu = (n_1, \dots, n_m)$ of $n$. In the context of Young's Rule (Section 4), it is helpful to visualize $\mu$ as a Young diagram, where each row corresponds to a homogeneous group of players. For example, if $n=5$ with groups of sizes $n_1=3$ and $n_2=2$:

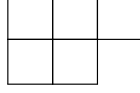
\begin{center}
    \ydiagram{3, 2}
\end{center}

A game $v$ is a \textit{homogeneous group game} if it is invariant under the action of $H$. Formally:
\begin{equation}
\mathcal{G}_{n}^{H}=\{v\in\mathcal{G}_{n}\mid\sigma v=v \text{ for all } \sigma\in H\}.
\end{equation}
In representation-theoretic terms, this is the isotypic component of the trivial representation of $H$ within the restriction of $\mathcal{G}_{n}$ to $H$.

For any $v\in\mathcal{G}_{n}^{H}$, the worth of a coalition $S$ depends only on the number of players drawn from each group \citep{Hernández2015}. We define the profile vector $s=(s_{1},\dots,s_{m})$ where $s_{j}=|S\cap T_{j}|$. The set of all valid profiles is the lattice $\Lambda$ in $\mathbb{Z}^{m}$:
\begin{equation}
\Lambda = \{(s_{1},\dots, s_{m}) \in \mathbb{Z}^m \mid 0 \le s_{j} \le n_{j} \}.
\end{equation}
\citet{Hernández2015} show that $v(S)$ is constant on the orbits of $H$ acting on $2^{N}$, and that each orbit corresponds uniquely to a profile $s\in\Lambda$. A basis for $\mathcal{G}_{n}^{H}$ is therefore given by the indicator functions of these orbits.

The dimension of this subspace is exactly the number of lattice points:
\begin{equation}
\dim(\mathcal{G}_{n}^{H})=\left[\prod_{j=1}^{m}(n_{j}+1)\right]-1.
\end{equation}
This reduction from $2^n$ dimensions to the polynomial size of $|\Lambda|$ is the key computational advantage identified by \citet{Hernández2015}, which we now analyze through the lens of induced representations.

\textit{Note on terminology:} Throughout this paper, the term ``group'' is used in two distinct but related contexts. When referring to the symmetric group $S_n$ or its Young subgroup $H$, it carries its standard algebraic meaning. When referring to subsets of players sharing identical attributes (e.g., $T_1, \dots, T_m$), it refers to categorical cohorts. The algebraic group $H$ is precisely the symmetry group preserving these categorical cohorts.

\section{Homogeneous Group Games and Their Space}
In this section, we characterize the space of homogeneous group games not merely as a collection of functions on a grid, but by identifying the specific irreducible representations of the symmetric group that constitute it.

\subsection{Characterization via Induced Representations}
Since $\mathcal{G}_{n}^{H}$ is defined by restriction to $H$, its spectrum of irreducible components is governed by induction to $S_{n}$ by Frobenius reciprocity.

\begin{thm} \label{thm:induced_rep}
An irreducible representation $V_\lambda$ of $S_n$ appears in the decomposition of the group games space $\mathcal{G}_n^H$ if and only if $V_\lambda$ is a constituent of both the induced representation $\mathrm{Ind}_H^{S_n}(\mathbf{1})$ and the base game space $\mathcal{G}_n$.
\end{thm}

\begin{proof} The space of homogeneous group games $\mathcal{G}_n^H$ consists of the $H$-invariant vectors in $\mathcal{G}_n$. An irreducible component $V_\lambda$ contributes to this subspace if and only if its restriction to $H$ contains a copy of the trivial representation $\mathbf{1}_H$. The dimension of the subspace of invariants in $V_\lambda$ is given by the intertwining number:
\begin{equation}
\dim \mathrm{Hom}_H(\mathbf{1}_H, \mathrm{Res}_H^{S_n}(V_\lambda)).
\end{equation}
By Frobenius Reciprocity, this is adjoint to induction:
\begin{equation}
\dim \mathrm{Hom}_H(\mathbf{1}_H, \mathrm{Res}_H^{S_n}(V_\lambda)) = \dim \mathrm{Hom}_{S_n}(\mathrm{Ind}_H^{S_n}(\mathbf{1}_H), V_\lambda).
\end{equation}
The right-hand side is non-zero if and only if $V_\lambda$ appears in the decomposition of $\mathrm{Ind}_H^{S_n}(\mathbf{1})$. Because $\mathcal{G}_n^H \subset \mathcal{G}_n$, any permissible representation must also be a constituent of the base game space.
\end{proof} 
The induced representation therefore acts as a ``spectral filter,'' selecting exactly which symmetry types are permissible in group games.

\subsection{The Isomorphism to the Profile Lattice}
The connection between Theorem \ref{thm:induced_rep} and the combinatorial profile lattice $\Lambda$ from Section 2 is central. Theorem \ref{thm:induced_rep} identifies which \textit{types} of symmetries are present, but the actual dimension of the space is determined by the orbit structure.

The space decomposes orthogonally by coalition size $k$:
\begin{equation}
\mathcal{G}_n^H = \bigoplus_{k=1}^n \mathcal{W}_k^H
\end{equation}
where $\mathcal{W}_k^H$ is the subspace of $H$-invariant games supported on coalitions of size $k$. The dimension of each $\mathcal{W}_k^H$ is exactly the number of solutions to $s_1 + \dots + s_m = k$ in the lattice $\Lambda$ \citep{Hernández2015}.

Since $\Lambda$ is graded by total coalition size, $\Lambda = \bigsqcup_{k=0}^{n} \{s \in \Lambda : s_1 + \cdots + s_m = k\}$, and the $k=0$ term is the single all-zero profile corresponding to the empty coalition, which is excluded from $\mathcal{G}_n$ by convention. Summing $\dim(\mathcal{W}_k^H)$ over $k = 1, \dots, n$ therefore accounts for every profile in $\Lambda$ except this one, recovering the dimension formula derived combinatorially by \citeauthor{Hernández2015}:
\begin{equation}
\dim(\mathcal{G}_n^H) = |\Lambda| - 1 = \left[ \prod_{j=1}^m (n_j + 1) \right] - 1.
\end{equation}
The next section applies Young's Rule to Theorem \ref{thm:induced_rep} to explicitly identify the partitions $\lambda$ that occur in $\mathcal{G}_n^H$, supplying the representation-theoretic reason for this complexity reduction.

\section{The Induced Decomposition and Subspace Intersection}

The central question is which irreducible components of $\mathcal{G}_n$ contribute to the symmetric subspace $\mathcal{G}_n^H$. $\mathcal{G}_n^H$ is not a representation of $S_n$---the full group does not preserve invariance under $H$---rather, it identifies which $S_n$-invariant subspaces $V_\lambda \subset \mathcal{G}_n$ contain non-zero $H$-invariant vectors.

\subsection{The Spectrum of Group Games}
Let $\mathcal{G}_{n}=\bigoplus_{\lambda\vdash n}m_{\lambda}V_{\lambda}$ be the decomposition of the full game space into irreducible $S_n$-modules $V_{\lambda}$, where $m_{\lambda}$ is the multiplicity. By Theorem \ref{thm:induced_rep}, a component $V_{\lambda}$ contributes to the group game space if and only if its restriction to $H$ contains the trivial representation.
Let $M^\mu$ be the induced representation module from the Young subgroup $H = S_{n_1} \times \cdots \times S_{n_m}$:
\begin{equation}
M^\mu := \mathrm{Ind}_{H}^{S_{n}}(\mathbf{1}).
\end{equation}
Here, the partition $\mu = (n_1, \dots, n_m)$ corresponds to the group sizes.

It follows that the ``spectrum'' of irreducible representations present in homogeneous group games is exactly the set of constituents of the permutation module $M^\mu$ that also naturally exist within $\mathcal{G}_n$.

\subsection{Young's Rule and the Minority Bound}
\label{sec:minority_bound}

The decomposition of the induced module $M^\mu$ is governed by Young's Rule (a specific case of the Littlewood-Richardson rule). For a Young subgroup $H = S_{n_1} \times \dots \times S_{n_m}$, the multiplicity of $V_\lambda$ in $M^\mu$ is the Kostka number $K_{\lambda \mu}$, which counts the number of semistandard Young tableaux of shape $\lambda$ and content $\mu$ \citep{Sagan2001}.

As established by \citet{Hernández2007}, the full game space $\mathcal{G}_n$ contains only two-part partitions $V_{[n-d, d]}$ for $0 \le d \le \lfloor n/2 \rfloor$. A fundamental property of Kostka numbers then places a strict bound on which of these components survive in group games.

\begin{thm} \label{thm:minority_bound}
For a two-group game ($m=2$), an irreducible representation $V_{[n-d, d]}$ contains an $H$-invariant vector (and thus appears in $\mathcal{G}_n^H$) only if the length of the second row satisfies $d \le n_{\min}$, where $n_{\min} = \min(n_1, n_2)$ is the size of the minority group.
\end{thm}

\begin{proof}
To form a valid semistandard Young tableau of shape $[n-d, d]$ with content $(n_1, n_2)$, entries must strictly increase down columns. The $d$ boxes in the second row must therefore be filled entirely with the index of the minority group (assuming without loss of generality that the minority group corresponds to index 2). Since the content only provides $n_{\min}$ copies of this index, the tableau can only be constructed if $d \le n_{\min}$. Simultaneously, the entries directly above in the first row are forced to carry the majority group index; this requires $d \le n_{\text{majority}}$, which is automatically satisfied since $d \le n_{\min} \le n_{\text{majority}}$. When $d \le n_{\min}$, the tableau is unique, yielding a multiplicity of exactly 1.
\end{proof}

This proves that the algebraic complexity of a group game with two groups is strictly bounded by the size of its minority group, mathematically collapsing the dimensions of the strategic interactions.

\subsection{Intersection with the Canonical Subspaces}
We now map this filtration onto the canonical dissection $\mathcal{G}_n = C \oplus U \oplus W$ established by \citet{Hernández2007}. Since the full game space is graded by coalition sizes $k \in \{1, \dots, n\}$, we analyze the multiplicities across all $n$ slices of the space.

Before analyzing each canonical subspace, we establish the multiplicity of the standard representation in each coalition-size slice. Here $\mathcal{W}_k$ denotes the (unrestricted) permutation module on $k$-element subsets of $N$, distinct from its $H$-invariant subspace $\mathcal{W}_k^H$ from Section~3.2.

\begin{lemma} \label{lem:multiplicity_one}
For each $0 < k < n$, the permutation module $\mathcal{W}_k$ contains the standard representation $V_{[n-1,1]}$ with multiplicity exactly $1$.
\end{lemma}

\begin{proof}
The symmetric group $S_n$ acts transitively on the set $\binom{N}{k}$ of $k$-element subsets of $N$; the stabilizer of any single subset is $S_k \times S_{n-k}$. By the orbit-stabilizer correspondence, the permutation module on $k$-subsets is the induced representation
\begin{equation}
\mathcal{W}_k \cong \mathrm{Ind}_{S_k \times S_{n-k}}^{S_n} \mathbf{1}.
\end{equation}
Applying Frobenius reciprocity,
\begin{equation}
\langle V_{[n-1,1]}, \mathcal{W}_k \rangle_{S_n} = \langle \mathrm{Res}_{S_k \times S_{n-k}} V_{[n-1,1]}, \mathbf{1} \rangle_{S_k \times S_{n-k}}.
\end{equation}
To evaluate the right-hand side, recall that $\mathbb{R}^n = V_{[n]} \oplus V_{[n-1,1]}$ as $S_n$-modules. Restricting $\mathbb{R}^n$ to $S_k \times S_{n-k}$, the standard basis $\{e_1, \dots, e_n\}$ splits into two orbits, yielding exactly two linearly independent invariant vectors (the partial sums $\sum_{i \le k} e_i$ and $\sum_{i > k} e_i$). Hence
\begin{equation}
\langle \mathrm{Res}_{S_k \times S_{n-k}} \mathbb{R}^n, \mathbf{1} \rangle_{S_k \times S_{n-k}} = 2.
\end{equation}
Since $V_{[n]}$ contributes one copy of the trivial representation, we subtract to obtain
\begin{equation}
\langle \mathrm{Res}_{S_k \times S_{n-k}} V_{[n-1,1]}, \mathbf{1} \rangle_{S_k \times S_{n-k}} = 2 - 1 = 1.
\end{equation}
\end{proof}

\subsubsection{The Constant Space $C$}
This corresponds to the trivial partition $\lambda = [n]$ (where $d=0$, as in Section~\ref{sec:minority_bound}). Since $0 \le n_{\min}$, it is always fully $H$-invariant. It appears exactly once for each coalition size, meaning the space of symmetric games $C$ has dimension $n$. Its intersection with group games is complete: $\dim(\mathcal{G}_n^H \cap C) = n$.

\subsubsection{The Essential Space $U$}
This corresponds to the standard partition $\lambda = [n-1, 1]$ (where $d=1$). In the full game space $\mathcal{G}_n$, the standard representation appears exactly $n-1$ times (once for each $0 < k < n$, by Lemma~\ref{lem:multiplicity_one}). By Young's Rule, assuming $n_{\min} \ge 1$, the multiplicity of $V_{[n-1,1]}$ in the induced module $M^\mu$ is the Kostka number $K_{[n-1,1],\mu} = m-1$. This means that within a \textit{single} copy of the standard representation, the $H$-invariant subspace has dimension $m-1$. Because there are $n-1$ copies total, the essential space for group games has dimension:
\begin{equation}
\dim(\mathcal{G}_n^H \cap U) = (n-1)(m-1).
\end{equation}

\subsubsection{The Null Space $W$}
The kernel $W$ consists of all $V_\lambda$-isotypic spaces with $\lambda \notin \{[n], [n-1, 1]\}$. In the general game space, $W$ contains deep interaction components up to $d = \lfloor n/2 \rfloor$. Theorem \ref{thm:minority_bound} shows that group symmetry severely truncates this: in the UN Security Council ($n=15$, $n_{\min}=5$), the general space has interactions up to $d=7$, but group symmetry algebraically prohibits the $d=6$ and $d=7$ components. The group-symmetric kernel is restricted to the ``shallow'' components of $W$ (specifically $[n-d, d]$ for $2 \le d \le n_{\min}$).

This structural filtration explains why symmetric values are so effective on group games: the complex higher-order interactions that usually plague solution concepts are algebraically ruled out by the constraints of the symmetry group $H$.

\section{Symmetric Linear Values and Axioms}

The decomposition derived in the previous section provides more than a structural classification; it dictates the behavior of any fair solution concept. We prove that any symmetric linear value on group games is strictly constrained by the invariant subspaces $C$, $U$, and $W$. Re-reading classical axioms through this algebraic lens reveals that ``fairness'' is mathematically equivalent to fixing specific scalars on the irreducible components of the game.

\subsection{The Vanishing of the Kernel $W$}
A central result in the representation theory of cooperative games is that the ``essential'' information for symmetric values is contained entirely within the isotypic components of irreducibles occurring in the natural representation of $S_n$ on $\mathbb{R}^n$ (namely $C$ and $U$).

\begin{prop}\label{prop:vanishing_kernel}
Let $\psi: \mathcal{G}_n \to \mathbb{R}^n$ be any linear symmetric value. Then for any game $w \in W \cap \mathcal{G}_n^H$, we have $\psi(w) = \mathbf{0}$.
\end{prop}

\begin{proof} As established by \citet{Hernández2007}, the space of all linear symmetric values is isomorphic to $\mathrm{Hom}_{S_n}(\mathcal{G}_n, \mathbb{R}^n)$. The target space $\mathbb{R}^n$ decomposes into only two irreducibles: the trivial representation $V_{[n]}$ and the standard representation $V_{[n-1, 1]}$. The subspace $W$ consists of irreducible components $V_\lambda$ where $\lambda \notin \{[n], [n-1, 1]\}$. By Schur's Lemma, there are no non-zero equivariant maps between non-isomorphic irreducible representations \citep{Fulton1991}. Consequently, any symmetric $\psi$ must map every component of $W$ to zero in $\mathbb{R}^n$.
Since group-symmetric games form a subspace $\mathcal{G}_n^H \subset \mathcal{G}_n$, the restriction of $\psi$ to the higher-order interaction components ($d \ge 2$ in $W$) remains zero. 
\end{proof}

This result implies that for two-group games, any interaction effects exceeding the structural bounds of the minority group are ``invisible'' to symmetric solutions. They represent pure ``noise'' filtered out by the requirement of fairness \citep{Kleinberg1985}.

\subsection{Scalar Action on $C$ and $U$}
With $W$ eliminated, a symmetric value on $\mathcal{G}_n^H$ is fully determined by its action on the intersection of the group game space with $C$ and $U$.

\begin{itemize}
    \item \textbf{Action on $C$:} The constant space consists of $n$ copies of the trivial representation, one for each coalition size. A symmetric linear value therefore requires $n$ independent parameters $(\alpha_1, \dots, \alpha_n)$ to specify its behavior on $C$ across the $n$ coalition sizes.
    \item \textbf{Action on $U$:} The full essential space $U$ consists of $n-1$ copies of the standard representation. Since $\mathbb{R}^n \cong V_{[n]} \oplus V_{[n-1,1]}$, Schur's Lemma implies $\mathrm{Hom}_{S_n}(V_{[n-1,1]}, \mathbb{R}^n)$ is one-dimensional, spanned by the canonical inclusion of $V_{[n-1,1]}$ into $\mathbb{R}^n$. Hence, on each of the $n-1$ isomorphic copies of $V_{[n-1,1]}$ in $U$, the restriction of $\psi$ is a scalar multiple of this fixed intertwiner, requiring $n-1$ parameters $(\beta_1, \dots, \beta_{n-1})$ across coalition sizes \citep{Hernández2007}.
\end{itemize}

Crucially, while the number of parameters defining the \textit{value} remains dependent on $n$, the \textit{game vector} it acts upon is severely restricted. Within each $k$-sized slice of the game, the essential strategic information is confined to the $m-1$ invariant dimensions. This shows that calculating the value requires tracking only how the $n-1$ scalar parameters scale the $(n-1)(m-1)$ essential degrees of freedom, rather than the exponential dimensions of the full space \citep{Dubey1981, Amer2003}.

\subsection{Axioms as Algebraic Constraints}
Standard game-theoretic axioms translate directly into linear constraints on these scalars.

\subsubsection{Efficiency}
The efficiency axiom requires $\sum \psi_i(v) = v(N)$. Because the constant space $C$ is parameterized by $n$ independent scalars (one for each coalition size), efficiency alone does not uniquely determine the value on $C$; it imposes a single linear constraint on these $n$ parameters, ensuring the total payoff sums to $v(N)$. Efficiency places no constraint on the components in $U$ or $W$: by construction, these subspaces are orthogonal to $C$, so for any $v \in U \oplus W$, the total worth of the grand coalition vanishes ($v(N)=0$). This reflects the standard game-theoretic concept of ``0-normalization,'' where total value is separated from the strategic distribution \citep{Hernández2007}.

\subsubsection{The Null Player Property}
A null player is a player that contributes zero marginal value. In the language of vector spaces, the null player axiom forces the solution to respect the internal structure of the standard representation. Applying it to the unanimity games $e_S$ for each coalition size $|S|=k$ yields a system of linear equations. Because a null player in $e_S$ contributes zero to all coalitions not containing $S$, the value assigns them zero payoff. Since the unanimity game satisfies $e_S(N) = 1$ by definition, efficiency then forces the remaining, non-null contributors to $e_S$ to split a total payoff of exactly $1$ among themselves --- this is what pins each $\beta_k$ to the specific value $1$, rather than to an arbitrary common constant. As demonstrated by \citet{Hernández2007}, enforcing this condition algebraically forces all $n-1$ scalar parameters $\beta_k$ on $U$ to equal $1$, while simultaneously fixing the remaining $n-1$ degrees of freedom in $C$.

\subsubsection{Uniqueness of the Shapley Value}
Imposing linearity, efficiency, and the null player property removes all $2n-1$ degrees of freedom across $C$ and $U$. The unique solution remaining is the Shapley value. Relaxing efficiency yields the broader class of \textit{semivalues} (such as the Banzhaf value), which can be understood as ``weighted'' projections onto the irreducible components of $\mathcal{G}_n^H$ \citep{Dubey1981}.

\section{Applications and Case Studies}

We now apply the algebraic decomposition to two distinct classes of problems: weighted voting in political bodies and the formation of economic partnerships. Both follow the general approach of \citet{Hernández2015}, who carried out a similar analysis, and they illustrate concretely how group symmetry simplifies the treatment of power and value distributions.

\subsection{Political Power: The UN Security Council}
We first consider the United Nations Security Council (UNSC). The council has $n=15$ members partitioned into two homogeneous groups:
\begin{itemize}
    \item \textbf{Group 1 (Permanent Members):} $n_1 = 5$ members (China, France, Russia, UK, USA) with veto power. (This serves as our minority group, $n_{\min} = 5$.)
    \item \textbf{Group 2 (Non-Permanent Members):} $n_2 = 10$ rotating members without veto power.
\end{itemize}
A resolution passes if it receives at least 9 votes in total, including all 5 votes from Group 1. The game $v_{UN}$ is defined by $v(S) = 1$ if the coalition $S$ is winning, and $0$ otherwise. Here $v_{UN}$ is a single, concrete element of the space $\mathcal{G}_{15}^H$ characterized in Sections 3--5; the framework developed there applies to \emph{every} game in that space, and we use $v_{UN}$ specifically to show how a symmetric value is computed once the profile lattice is in hand.

The symmetry group is $H = S_5 \times S_{10}$. The dimension of the full game space is $2^{15} - 1 = 32{,}767$, while the dimension of the $H$-invariant subspace $\mathcal{G}_{15}^H$ is only $(5+1)(10+1) - 1 = 65$. This reduction allows us to compute the Shapley value using the lattice coordinates alone.

We analyze the critical boundaries within the profile lattice $\Lambda$. A non-permanent member is pivotal only at the specific lattice transition from $s=(5,3)$ to $s=(5,4)$, where the coalition has all 5 veto holders and exactly 3 non-permanent members. In contrast, a permanent member is pivotal at any transition from $(4, s_2)$ to $(5, s_2)$ for $s_2 \in \{4, \dots, 10\}$.

Summing the probability weights of these lattice paths gives the exact Shapley values: $\phi_{perm} \approx 0.1963$ and $\phi_{non-perm} \approx 0.0018$. Permanent members constitute only $33\%$ of the council, yet the structural rigidity of the veto grants them $98.15\%$ of the total voting power.

Decomposing the game into its $C$ and $U$ components recovers the classical result that veto power grants each permanent member approximately $19.6\%$ of total power, while non-permanent members hold only $\approx$ $0.18\%$ each \citep{Hernández2015}. The ``interaction'' components in $W$ reveal that the veto structure suppresses almost all higher-order synergies between non-permanent members.

\subsection{Economic Synergy: The Market for Complementary Goods}
A more ubiquitous application arises in supply chains and assembly problems, classically modeled as the ``Glove Game.'' Consider a market with $n$ agents in two groups:
\begin{itemize}
    \item \textbf{Group L (Left Holders):} $n_L$ agents, each holding one unit of a left-hand component (e.g., a chassis).
    \item \textbf{Group R (Right Holders):} $n_R$ agents, each holding one unit of a right-hand component (e.g., an engine).
\end{itemize}
A unit of economic worth is generated only when a complete pair is assembled. The characteristic function is:
\begin{equation}
v(S) = \min(|S \cap L|, |S \cap R|).
\end{equation}
This game is invariant under $H = S_{n_L} \times S_{n_R}$. Unlike the voting game, which is defined by a threshold, this game is driven by strictly complementary matching.

\subsubsection{Decomposition Analysis}
If $n_L < n_R$, the left holders are the ``scarce'' resource. In the canonical decomposition $\mathcal{G}_n = C \oplus U \oplus W$:
\begin{enumerate}
    \item The constant component $C$ may carry the largest share of the game vector's raw magnitude (in the $n=4$ example below, over 90\%), but it cannot, by itself, distinguish Group L from Group R: it is fixed entirely by coalition size, and therefore provides an identical baseline payoff to every player regardless of group.
    \item The essential component $U$ (the standard representation) carries the weight of the scarcity. The Shapley value assigns significantly higher payoff to Group L members. More precisely, $U$ carries the scarcity signal because the projection into $C$ distributes an identical baseline payoff to all players (as it lies in the strictly $S_n$-symmetric part of the space); therefore, any difference in Shapley payoff between Group L and Group R is, without exception, attributable entirely to the projection onto $U$.
    \item The kernel component $W$ (along with parts of $U$) captures the ``waste.'' In general, any coalition with excess right gloves has worth equal to its count of left gloves, since $v(S) = \min(|S \cap L|, |S \cap R|) = |S \cap L|$ whenever $|S \cap L| < |S \cap R|$; a coalition of, say, 3 Left and 10 Right gloves has worth $v(S) = 3$, and the 7 excess right gloves add nothing further to that worth, even though they still contribute to the game's components in $U$ and $W$. Symmetric values explicitly nullify the $W$ contribution: by Proposition~\ref{prop:vanishing_kernel}, any symmetric linear value $\psi$ vanishes identically on $W$, so the excess-glove imbalance captured there has no effect on any player's payoff --- only its shadow in $C \oplus U$ (the raw glove counts and their leading-order imbalance) actually determines the final Shapley payoff.
\end{enumerate}
Our representation-theoretic framework allows us to locate the exact coordinates of this market by examining the algebraic structure of the game vector itself. We can rewrite the matching function exactly as:
\begin{equation}
v(S) = \min(|S \cap L|, |S \cap R|) = \frac{1}{2}|S| - \frac{1}{2}\big| |S \cap L| - |S \cap R| \big|.
\end{equation}
This algebraic formulation elegantly maps onto our canonical decomposition. The first term, $\frac{1}{2}|S|$, is strictly a function of coalition size; it therefore resides entirely within the constant space $C \cong nV_{[n]}$, where the coefficient $n$ denotes the multiplicity of the trivial representation $V_{[n]}$ in $\mathcal{G}_n$ (recall from Section~4 that $V_{[n]}$ appears once for each of the $n$ coalition sizes). This term provides a uniform baseline payoff to all players, blind to their group affiliation. The strategic asymmetry of the market is isolated within the second term, $-\frac{1}{2}\big| |S \cap L| - |S \cap R| \big|$. We call this term \emph{non-linear} because, unlike the first term, it cannot be written as a single fixed linear combination of $|S \cap L|$ and $|S \cap R|$: it equals $|S \cap L| - |S \cap R|$ on the region $|S \cap L| \ge |S \cap R|$, and $|S \cap R| - |S \cap L|$ on the region $|S \cap L| \le |S \cap R|$, two different linear formulas depending on which group holds the excess.

The second term does not project into a single canonical subspace; rather, it splits across $U$ and $W$. Its projection onto $U$ (spanned by copies of $V_{[n-1,1]}$) extracts exactly the part of $\big||S \cap L| - |S \cap R|\big|$ that is linear in the coalition's group composition --- the closest linear approximation to a function that, by its very construction from an absolute value, cannot be linear everywhere. Applying $\phi$ to this component yields a positive payoff shift for each minority-group player and a negative payoff shift for each majority-group player, algebraically counteracting the uniform baseline $\phi(v_C)$ awarded by $C$. Only the part of the second term that \emph{cannot} be captured by any linear formula in $|S \cap L|, |S \cap R|$ --- precisely the fact that it is given by different linear formulas on the two regions $|S \cap L| \ge |S \cap R|$ and $|S \cap L| \le |S \cap R|$ of the profile lattice $\Lambda$, with no single formula valid on both --- is left over for $W \cap \mathcal{G}_n^H$, which has dimension $(n_{\min}-1)(n-n_{\min}-1)$: one dimension contributed by each subrepresentation isomorphic to $V_{[n-k,k]}$ for $2 \le k \le n_{\min}$, since by Theorem~\ref{thm:minority_bound} only a single $H$-fixed line survives from each of its (generally many) copies inside $\mathcal{G}_n$. (When $n_L = n_R$, swapping the two groups is an actual symmetry of the game that negates $|S\cap L|-|S\cap R|$ while fixing its absolute value, forcing the $U$-projection to vanish identically in that special case; in general, for $n_L \ne n_R$, both the $U$- and $W$-projections of the second term are nonzero.)

Because Proposition~\ref{prop:vanishing_kernel} dictates that $\psi(W) = \mathbf{0}$, the algebraic filtration proves a remarkable property of symmetric values: fairness is provably blind to the piecewise matching-threshold structure of the market. Any symmetric linear value depends on that structure only through its linear part, captured in $U$; the genuinely nonlinear part, captured in $W$, contributes nothing to any player's payoff.

\subsubsection{Explicit Computation of the Kernel for $n=4$}
To make the abstract projection into $W$ concrete, we explicitly compute the kernel for a small symmetric market. Let $n=4$, partitioned perfectly into $n_L=2$ and $n_R=2$. Let the player set be $N = \{L_1, L_2, R_1, R_2\}$. 

We isolate the ``slice'' of the game acting on coalitions of size $k=2$. There are $\binom{4}{2} = 6$ such coalitions. Four of these coalitions are ``mixed'' (one left, one right glove) and evaluate to a coalition worth of $1$. Two of these coalitions are ``homogeneous'' (both left, or both right) and evaluate to a coalition worth of $0$. The game vector for this slice is:
\begin{equation}
v_2 = (1, 1, 1, 1, 0, 0)
\end{equation}

We project $v_2$ into the canonical subspaces. The constant space $C$ calculates the symmetric baseline, which is simply the average worth of a 2-player coalition: $c_2 = \frac{4}{6} = \frac{2}{3}$. Therefore, $v_{2, C} = (\frac{2}{3}, \frac{2}{3}, \frac{2}{3}, \frac{2}{3}, \frac{2}{3}, \frac{2}{3})$. Furthermore, because the sizes of the two groups are perfectly equal ($n_L = n_R$), there is no relative scarcity in the market. Consequently, the essential component vanishes entirely: $v_{2, U} = \mathbf{0}$.

The residual data lies strictly in the null space $W$. We compute this by subtraction, $v_{2, W} = v_2 - v_{2, C}$:
\begin{itemize}
    \item For the four mixed coalitions, the $W$ component is $1 - \frac{2}{3} = \frac{1}{3}$.
    \item For the two homogeneous coalitions, the $W$ component is $0 - \frac{2}{3} = -\frac{2}{3}$.
\end{itemize}
The Shapley formula for $\phi_i(v_{2,W})$ involves only two types of transitions, since $v_{2,W}$ vanishes on every coalition of size $1$, $3$, and $4$: player $i$ joining a singleton to form a size-2 coalition, and a size-2 coalition not containing $i$ growing to size $3$. Each transition is weighted by $\frac{k!(n-k-1)!}{n!} = \frac{1}{12}$ at $k \in \{1,2\}$, $n=4$. For any fixed player $i$, exactly two of the three size-2 coalitions containing $i$ are mixed and one is homogeneous, contributing $2(\frac{1}{3}) + 1(-\frac{2}{3}) = 0$; by the same count, exactly two of the three size-2 coalitions excluding $i$ are mixed and one is homogeneous, contributing $-\big(2(\frac{1}{3}) + 1(-\frac{2}{3})\big) = 0$. Both transitions vanish, so $\phi_i(v_{2,W}) = 0$ for every player $i$. This explicitly demonstrates how the non-linear matching thresholds of the Glove Game populate the $W$ space, and how the Shapley value mathematically filters them out.

\subsection{Computational Complexity}
The most practical implication of our result is computational. Calculating the Shapley value strictly via permutations requires $O(n!)$ operations, though this can be reduced to $O(2^n)$ by summing over coalitions. For a group-symmetric game with $m$ groups, our representation-theoretic framework allows us to compute values by summing only over the lattice points $\Lambda$.

The complexity reduces to:
\begin{equation}
O(n^m).
\end{equation}
For the UNSC example ($m=2$), this reduces the operation count from $\approx 3 \times 10^4$ to $\approx 65$. For a supply chain with 3 types of components and 100 agents, the reduction is from astronomical ($2^{100}$) to manageable ($100^3$), making the Shapley value a feasible metric for large-scale economic systems.

\section{Conclusion}

In this paper, we develop a rigorous algebraic framework for analyzing cooperative games with homogeneous groups of players. Rather than treating symmetry as a simplifying assumption, we show it is an intrinsic algebraic feature of the game space, and that the polynomial complexity results of \citet{Hernández2015} are direct consequences of the underlying module structure.

\subsection{Summary of Results}
By treating symmetry not as a simplifying assumption but as an algebraic feature, we gain the ability to dissect complex strategic interactions into their fundamental atomic components. Our primary contribution is the characterization of the spectrum of group-symmetric games $\mathcal{G}_n^H$ using the induced representation $\mathrm{Ind}_H^{S_n}(\mathbf{1})$. Using Frobenius Reciprocity and Young's Rule, we prove that, for two-group games, this space decomposes into irreducible representations of $S_n$ corresponding only to partitions where the depth of interactions is bounded by the size of the minority group. This structural constraint has the following implications for solution concepts:

\begin{enumerate}
    \item \textbf{Filtration of Complexity:} For two-group games ($m$=2), the ``deep'' kernel components of $W$ (irreducibles $V_{[n-d, d]}$ where $d > n_{\min}$, and $n_{\min}$ is the amount of players in the smallest sized group of players) are algebraically prohibited. This explains why symmetric values on such games act upon a highly restricted set of geometric coordinates.
    \item \textbf{Axiomatic Rigidity:} Classical axioms like efficiency and the null player property act as linear constraints on the surviving irreducible components ($C$ and $U$). For $m=2$, assuming linearity, these constraints leave no degrees of freedom, rendering the Shapley value the unique symmetric solution.
    \item \textbf{Computational Reduction:} The dimension of the invariant subspace grows polynomially as $O(n^m)$, validating the efficiency of lattice-based algorithms for value computation.
\end{enumerate}

\subsection{Future Directions}
Several avenues remain open. Our analysis focuses on the Young subgroup $H$, which corresponds to a specific flat partition of players. A natural extension is to \textbf{wreath product symmetries} $S_k \wr S_m$, which would model games with hierarchical structures (e.g., groups of groups), typical in federal governance systems.

Second, while we focused on linear values, this decomposition could be applied to \textbf{games in partition function form}, where externalities exist between coalitions. Recent work has suggested that representation theory can characterize solutions in these more complex environments \citep{Sanchez2014}, and the induced module approach provides a natural pathway to formalize ``group externalities.''

Finally, since $W$ is precisely the space of components invisible to every linear symmetric value (Proposition~\ref{prop:vanishing_kernel}), it behaves formally like a coboundary space in a suitable cochain complex on the coalition lattice; making this analogy precise --- e.g., by relating $W$ to the kernel of a simplicial coboundary map on $2^N$ --- could connect solution-concept theory to the cohomology of game spaces. The components of $W$ annihilated by symmetric values may carry information relevant to non-cooperative stability or core non-emptiness. Investigating these ``hidden'' subspaces using Specht modules could yield new classes of invariants for cooperative games.

\section*{Acknowledgments}
We thank Professor Joshua Lansky of American University for his invaluable guidance and mentorship throughout the development of this paper.

\bibliographystyle{apalike}
\bibliography{references}

\end{document}